\documentclass[aps,nofootinbib,preprintnumbers]{revtex4-1}
\usepackage{color}

\usepackage{amsthm,amsmath,amssymb,amscd}
\usepackage{amssymb}
\usepackage{graphics}

\usepackage{slashed}
\usepackage{amssymb}
\usepackage{braket}
\usepackage{upgreek}

\newcommand{\half}{\frac{1}{2}}

\newcommand{\beq}{\begin{eqnarray}}
\newcommand{\eeq}{\end{eqnarray}}

\def\a{\alpha}

\usepackage{tikz}

\newtheorem{thm}{Theorem}[section]
\newtheorem{prop}[thm]{Proposition}

\newtheorem{lemma}[thm]{Lemma}
\newtheorem{defn}[thm]{Definition}

\newtheorem{preremark}[thm]{Remark}
\newenvironment{remark}{\begin{preremark}\rm}{\medskip \end{preremark}}
\numberwithin{equation}{section}

\begin{document}
\title{ Anomalous Dimensions for Boundary Conserved Currents in Holography via the Caffarelli-Silvestre Mechanism for p-forms}


\author{Gabriele La Nave}
\affiliation{Department of Mathematics, University of Illinois, Urbana, Il. 61820}
\author{Philip W. Phillips}
\affiliation{Department of Physics and Institute for Condensed Matter Theory, University of Illinois, 1110 W. Green Street, Urbana, IL 61801}

\begin{abstract}
Although it is well known that the Ward identities prohibit anomalous dimensions for conserved currents in local field theories, a claim from certain holographic models involving bulk dilaton couplings is that the gauge field associated with the boundary current can acquire an anomalous dimension.  We resolve this conundrum by showing that all the bulk actions that produce anomalous dimensions for the conserved current generate non-local actions at the boundary.  In particular, the Maxwell equations are fractional.  To prove this, we generalize to p-forms the Caffarelli/Silvestre (CS) extension theorem.  In the context of scalar fields, this theorem demonstrates that second-order elliptic differential equations in the upper half-plane in ${\mathbb R}_+^{n+1}$ reduce to one with the fractional Laplacian, $\Delta^\gamma$, with $\gamma \in \mathbb R$, when one of the dimensions is eliminated.   From the p-form generalization of the CS extension theorem, we show that at the boundary of the relevant holographic models, a fractional gauge theory emerges with equations of motion of the form, $\Delta ^\gamma A^t= 0$ with $\gamma$ $\in R$ and $A^t$ the boundary components of the gauge field.  The corresponding field strength $F = d_\gamma A^t= d\Delta ^\frac{\gamma-1}{2} A^t$ is invariant under 
$A^t \rightarrow A^t+ d_\gamma \Lambda$ with the fractional differential given by
 $d_\gamma \equiv (\Delta)^\frac{\gamma-1}{2}d$, implying that $[A^t]=\gamma$ which is in general not unity.
  \end{abstract}

\maketitle

\section{Introduction}

A text-book problem\cite{gross,peskin} in quantum field theory is to prove that conserved quantities such as the electrical current cannot acquire anomalous dimensions even under renormalization.  The basic argument is that the current, $J_\mu$, enters the action through the combination $J_\mu A^\mu$, where $A_\mu$ is the vector potential.  As long as the theory remains gauge-invariant, the transformation $A_\mu\rightarrow A_\mu+\partial_\mu\Lambda$, ensures that $[A_\mu]=1$ and hence the dimension of $J_\mu$ is fixed by the volume factor in the action; that is, $[J_\mu]=d-1$\cite{gross,peskin}.  Nonetheless, holographic\cite{g1,g2} bulk models have been constructed in which the dimension of the current at the boundary is arbitrary and hence so is the dimension of associated gauge field, $A_\mu$.   Since $A$ is a differential 1-form, precisely what it means for it to acquire an anomalous dimension is unclear.  In this note, we show that the operative mechanism for changing the dimension of $A$ in the extant holographic constructions \cite{g1,g2} is the p-form generalization of the Caffarelli-Silvestre\cite{CS2007} mechanism. The conformal version of this theorem can be found in the works of Graham and Zworski \cite{graham} and Chang and Gonzalez \cite{chang:2010}.  What we prove here is that although the dual theory\cite{klebanov,Witten1998} is governed by currents that in principle do not obey the standard local gauge group, they are controlled by a fractional gauge group in which $A \rightarrow A+ d_\gamma \Lambda$, where $d_\gamma=(\Delta)^\frac{\gamma-1}{2}d$, with $\gamma \in \mathbb R$ and $\Lambda \in C^2(M)$.  We provide an explicit proof of this here. 

That there is a fundamental connection between the Caffarelli/Silvestre mechanism and the holographic models that generate anomalous dimensions for the gauge field is ultimately not surprising given that such models are all based on dilaton actions of the form, 
\beq
S=\int d^{d} x dy\sqrt{-g}Z(\phi)F^2+\cdots,
\eeq
where $F$ is the field strength and $y$ is the radial direction.  The class of solutions\cite{g1,g2} that yields the anomalous dimension for the gauge field has the dilaton field scaling as $Z(\phi)\propto y^a$.   Consequently, the equations of motion are equivalent to
\beq
\nabla^\mu( y^aF_{\mu\nu})=0.
\eeq
In the language of differential forms, this equation becomes
\beq
d(y^a\star dA)=0,
\eeq
which clearly illustrates that along any slice perpendicular to the radial direction, the standard $U(1)$ gauge transformation applies.
To determine what happens at the boundary, we note that
these equations are reminiscent of  those studied by Caffarelli and Silvestre\cite{CS2007} (CS)  for the case of a scalar field,
\beq
\nabla\cdot(y^a \nabla u),
\eeq
which is just a recasting of the (degenerate) elliptic differential equation
\beq
u(x,y=0)=f(x)\\
\Delta_x u+\frac{a}{y}u_y+u_{yy}=0.
\eeq
What they were interested in is what form does this differential equation acquire at the boundary, $y\rightarrow 0$.
They showed that any equation of this kind satisfies
\beq
\lim_{y\rightarrow 0^+}(-y^a u_y)=C_{d,\gamma}(-\Delta)^\gamma f(x).
\eeq
where $\gamma=(1-a)/2$ and $\Delta$ is the Euclidean Laplacian.  We show here that the same result holds for a differential p-form and hence the boundary action in holographic models that yield anomalous dimensions is of the form,
\beq\label{sfrac}
S =  \frac{1}{2}\int A_i (-\nabla)^{2\gamma} A^i,
\eeq
thereby giving rise to fractional scaling dimension for $A$.  The corresponding field strength is the 2-form,
\beq
F = d_\gamma A= d\Delta ^\frac{\gamma-1}{2} A,
\eeq
with gauge-invariant condition,
\beq\label{eq:u1frac}
A \rightarrow A+ d_\gamma \Lambda,
\eeq
 with 
\beq
 \ d_\gamma \equiv (\Delta)^\frac{\gamma-1}{2}d,
 \eeq
 which preserves the 1-form nature of the gauge-field with dimension $[A_\mu]=\gamma_\mu$, rather than unity. In the process we introduce various equivalent definitions for $d_\gamma$ above and show that mathematically it is the correct generalization of the standard differential of forms, in that (cf. Theorem \ref{frac-lap-diff})
 \beq 
 d_\gamma d_\gamma ^* +  d_\gamma^* d_\gamma = \Delta ^\gamma
 \eeq
 
 There is a precedent for fractional gauge groups in boundary actions.   In a spacetime that is asymptotically hyperbolic, Domokos and Gabadze\cite{domokos} considered a bulk action with $F^2$ with a gauge transformation of the standard form along the boundary coordinates but a fractional transformation,
 \beq
 A_y\rightarrow A_y+\partial_y^\gamma A_y\\
 \eeq
 along the radial direction.  Here $\partial_y^a$ is the fractional derivative.  They then integrated out the radial direction and obtained the fractional action, Eq. (\ref{sfrac}), for the boundary components of the gauge field.  The relationship between this mechanism and the dilaton approach is that asymptotically the gauge field has an algebraic form at the boundary.  Hence, powers of the radial coordinate can be substituted for derivatives.  As a result, fractional derivatives along y-coordiante will translate to a bulk coupling of the dilaton form.  

From a mathematical standpoint, our work opens up the possibility of studying new conformal invariants in the form of the fractional Paneitz operators introduced by Graham-Zworski \cite{graham} and Chang-Gonz\'alez \cite{chang:2010}. An interesting future direction is to analyze the main mathematical theorem proven here both in the $\gamma>1$ case and in the general manifold case. An important role here will be played by the analogue of the Bochner-Weitzenb\"ock formulae in the fractional setting, and we expect significant changes to occur already  at the $\gamma=\frac{1}{2}$ case, as seen in the function case in Chang-Gonz\'alez's work.
 We also expect fractional Laplacians on forms to be related to something analogous to the Paneitz operators in the conformally compact setting of which our fractional Laplacian will just be the highest order term, as in the work of  Graham and Zworski \cite{graham} and Chang and Gonzalez \cite{chang:2010}. This will be part of future research.
\section{The fractional Laplacian on forms}

Throughout the paper, following standard nomenclature, we will denote by $\Omega ^p(M)$ the space of ~p-forms on a manifold $M$. Let us fix the dimension of the manifold to be $n$.
We recall a few facts about the Laplacian on manifolds. First the Hodge star operator, $\star: \Omega ^p(M)\to \Omega ^{d-p}(M)$ which is defined by requiring that

$$\star \left( e_{i_1}\wedge \cdots  \wedge e_{i_p}\right) = e_{j_1} \wedge \cdots \wedge e_{j_{n-p}},$$
if $\left\{ e_{i_1}, \cdots , e_{i_p}, e_{j_1} , \cdots ,e_{j_{d-p}}\right\}$ is a {\it positive} frame of the cotangent bundle $T^*M = \Omega ^1(M)$.
The inner product on forms is defined by $(\alpha, \beta)= \int_M \langle \alpha , \beta \rangle \; \star 1$, where  $\langle \alpha , \beta \rangle $ is the pointwise scalar product on forms and $\star 1= dV$, the volume form. The scalar product  $(\alpha, \beta)$ is readily seen to equal
$$(\alpha, \beta)= \int_M \alpha \wedge \star \beta.$$
Also, recall that the adjoint of the differential operator, denoted by $d^*$, is an operator $d^*:  \Omega ^p(M)\to  \Omega ^{p-1}(M)$ which satisfies the defining property{\begin{footnote}{ Here and in the sequel we will be purposely vague about the nature of $M$, that is, whether it is closed or with boundary or compact or not.   We will simply assume that the space of forms we take is the subspace of the space of forms that is necessary for integration by parts to hold without boundary terms.}\end{footnote}}
$$(d\alpha, \beta) = (\alpha , d^*\beta).$$
It is a standard fact that the following holds
\begin{lemma}
One has that on $p$-forms, $d^*= (-1)^{n(p+1)+1}\; \star d\star$.
\end{lemma}
Let us recall that the Hodge Laplacian on $p$-forms is defined as
$$\Delta = dd^*+d^*d: \Omega ^p (M)\to  \Omega ^p (M).$$
One readily show that if $M=\mathbb R ^n$ with the standard flat metric $ds^2= \delta _{ij} \,  dx^i \otimes dx^j$, the calculation of $d^*d\omega$, $dd^*\omega$ and $\Delta \omega$ for a $p$-form $\omega= \omega _{i_1\cdots i_p} dx^{i_1} \wedge \cdots dx^{i_p}$ amounts to

\beq\label{eucld*}
d^*\omega = \sum _{\ell=1}^p \, (-1)^{(p+1)(2d-p)+\ell} \frac{\partial  \omega _{i_1\cdots i_p} } {\partial x^{i_\ell}}\, dx^{i_1} dx^{i_1} \wedge \cdots \wedge \hat { dx^{i_\ell} } \wedge \cdots\wedge dx^{i_p},
\eeq
whence

\beq \label{eucldd*}
\begin{aligned} dd^*\omega&= \sum _{\ell =1}^p (-1) \frac{\partial ^2  \omega _{i_1\cdots i_p} } {(\partial x^{i_\ell})^2}\, dx^{i_1} \wedge \cdots dx^{i_p}\\&+ \sum _{\ell =1}^p  \sum _{k =1}^{n-p}(-1)^{\ell}   \frac {\partial ^2  \omega _{i_1\cdots i_p} } {\partial x^{i_\ell} \partial x^{j_k}} \,
dx^{j_k} \wedge dx^{i_1} \wedge \cdots \wedge \hat { dx^{i_\ell} } \wedge \cdots\wedge dx^{i_p}, \end{aligned}\eeq

and
\beq \label{eucld*d}
\begin{aligned} d^*d\omega&= \sum _{k=1}^{n-p} (-1) \frac{\partial ^2  \omega _{i_1\cdots i_p} } {(\partial x^{j_k})^2}\, dx^{i_1} \wedge \cdots dx^{i_p}\\&+ \sum _{\ell =1}^p  \sum _{k =1}^{n-p}(-1)^{\ell +1}   \frac {\partial ^2  \omega _{i_1\cdots i_p} } {\partial x^{i_\ell} \partial x^{j_k}} \,
dx^{j_k} \wedge dx^{i_1} \wedge \cdots \wedge \hat { dx^{i_\ell} } \wedge \cdots\wedge dx^{i_p}. \end{aligned}\eeq
Putting equations \eqref{eucldd*} and \eqref{eucld*d} together yields

%

\beq \label{eucllap} \Delta = -\sum _{m=1}^n \frac{\partial ^2  \omega _{i_1\cdots i_p}}{(\partial x^m)^2} \, dx^{i_1} \wedge \cdots \wedge dx^{i_p}.
\eeq
This clearly shows that in general $\Delta$ is an elliptic operator. Furthermore, using the fact that $d^*$ is the adjoint of $d$, one can show that $\Delta$ is a symmetric operator,
$$(\Delta \alpha, \beta)= (\alpha , \Delta \beta).$$

\noindent
We thus define, following the spectral theorem, the fractional Laplacian on forms as
\beq \label{fract-lap-def}\Delta ^\gamma \alpha= \frac{1}{\Gamma (-\gamma)} \int _0^{\infty} \; \left( e^{-t \Delta} \alpha - \alpha\right) \frac{dt}{t^{1+\gamma}},\eeq
for $\gamma\in (0,1)$. For negative powers, we define
\beq\label{fractionalintegral}
\Delta ^{-s} \omega = \frac{1}{\Gamma(s)}\int_0^{+\infty} e^{-t\Delta} \omega \frac{dt}{t^{1-s}}
\eeq
with $s>0$.
and as one does for the fractional Laplacian on functions, we define
\beq
\Delta ^\gamma=\Delta ^{\gamma -\lfloor{\gamma}\rfloor} \Delta ^{\lfloor{\gamma}\rfloor}
\eeq
where $\lfloor{\gamma}\rfloor$ indicates the integral part of $\gamma$.
In fact, this makes sense for any self-adjoint operator and in particular it applies to both $dd^*$ and $d^*d$.
Here, the heat {\it semigroup} $e^{-t \Delta} \alpha$ on forms is defined by requiring that $e^{-t \Delta} \alpha$ be equal to the form $\beta$ which is the solution to the diffusion equation
\beq\label{semigroup}
\left\{ \begin{aligned}&\frac{\partial} {\partial t}  \beta + \Delta \beta = 0, \;\;\; \text{ for } \; (x,t) \in M\times \mathbb R_+\\& \beta (x,0) = \alpha (x)\;\;\; \text{ for } \; x\in M.\end{aligned} \right.
\eeq
Recall,  that such solutions always exist if $\omega$ is assumed sufficiently regular ($\omega \in C^{2,\alpha}$ is the correct requirement) is part of the famed theorem of Milgram-Rosenbloom\cite{j} ( see Theorem 3.6.1).
We next prove a sequence of important facts which will turn out to be useful for later purposes.
\begin{lemma}\label{lapbinomial}
For any $p$-form $\omega$, one has
$$\Delta ^\gamma \omega= \sum _{k=0}^\infty \, {{\gamma}\choose{k}} \, (dd^*)^ {\gamma-k} (d^*d)^k \omega $$
where ${{a}\choose{k}}$ are Newton's binomial coefficients.
\end{lemma}
\begin{proof}
This is a straightforward consequence of the definition and of Newton's binomial theorem.
\end{proof}
\begin{lemma}\label{lapdis0}
$\Delta ^\gamma d=(dd^*)^ \gamma d$.
\end{lemma}
\begin{proof}

Using Lemma \ref{lapbinomial} we can write
$$\Delta ^\gamma d \omega =  \sum _{k=0}^\infty \, {{\gamma}\choose{k}} \, (dd^*)^ {\gamma-k} (d^*d)^k d\omega  $$
but $ (d^*d)^k d\omega=0$ unless $k=0$, whence
$$\Delta ^\gamma d \omega =(dd^*)^ \gamma d\omega . $$
Alternatively, it follows from the definition \ref{fract-lap-def}, observing that if $\beta$ solves eq.\eqref{semigroup} then $d\beta$ solves
$$\left\{ \begin{aligned}&\frac{\partial} {\partial t}  d\beta + \Delta d\beta = 0, \;\;\; \text{ for } \; (x,t) \in M\times \mathbb R_+\\& d\beta (x,0) = d\alpha (x)\;\;\; \text{ for } \; x\in M.\end{aligned} \right.$$
and therefore 
$$e^{-t\Delta} d\beta = e^{-t(dd^*)}d\beta$$
\end{proof}

\begin{prop}\label{identities}
For any $a,b\in \mathbb R$
$$d (dd^*)^a =0,\qquad d ^* (d^*d)^a =0 \qquad \text{ and } \qquad [(d^*d)^b,(d^*d)^a ]=0$$
where $[\cdot, \cdot]$ denotes the commutator, and also
$$ (d^*d)^ad =0,\qquad (dd^*)^a d^*=0. $$

\end{prop}
\begin{proof}
According to the comments following the definition of the fractional Laplacian ( Eq. \eqref{fract-lap-def}), one defines
\beq \label{fractionalddstar} (dd^*)^a \omega= \frac{1}{\Gamma (-a)} \int _0^{\infty} \; \left( e^{-tdd^*} \omega - \omega\right) \frac{dt}{t^{1+a}},\eeq
where $\beta=e^{-tdd^*} \omega $ is the unique solution to

\beq\label{dd*diffusion}\left\{ \begin{aligned}&\frac{\partial} {\partial t}  \beta + dd^* \beta = 0, \;\;\; \text{ for } \; (x,t) \in M\times \mathbb R_+\\& \beta (x,0) = \omega (x)\;\;\; \text{ for } \; x\in M.\end{aligned} \right.\eeq
It then follows from this differential equation that $d \beta$ satisfies the equation
\beq\label{variationofd} \left\{ \begin{aligned}&\frac{\partial} {\partial t} (d \beta)  = 0, \;\;\; \text{ for } \; (x,t) \in M\times \mathbb R_+\\&d \beta (x,0) = d\omega (x)\;\;\; \text{ for } \; x\in M\end{aligned}, \right.\eeq
having used the fact that $d^2=0$ hence $ddd^* \beta = 0$. Clearly equation \eqref{variationofd} implies that $d(e^{-tdd^*}\omega ) = d\omega$ for any $t$ and therefore taking the differential of equation \eqref{fractionalddstar} implies that
$$d(dd^*)^a \omega= \frac{1}{\Gamma (-a)} \int _0^{\infty} \; \left( de^{-tdd^*} \omega - d\omega\right) \frac{dt}{t^{1+a}}=0.$$

Since the operator $dd^*$ is neither self-adjoint (its adjoint is $d^*d$) nor elliptic (although it is degenerate elliptic with degeneracy at those forms $\beta$ such that $d^*\beta=0$) the preceding argument needs justification. Specifically, we need to show that equation \eqref{dd*diffusion} has eternal (i.e., for any $t>0$) solutions. This is done as follows. Let $\eta$ be the (unique) solution to 
\beq\label{diffeq}\left\{ \begin{aligned}&\frac{\partial} {\partial t}  \eta + \Delta \eta = 0, \;\;\; \text{ for } \; (x,t) \in M\times \mathbb R_+\\& \eta (x,0) = \omega (x)\;\;\; \text{ for } \; x\in M.\end{aligned} \right.\eeq
Recall that the space of $p$-forms decomposes, according to the Kodaira decomposition as{\begin{footnote}{To be precise, in the decomposition of Eq. \eqref{hodge}, one should consider the $L^2$ closure $B_p$ and $B_p^*$ of the spaces $d\Omega ^{p-1}(M)$ and $ d^* \Omega ^{p+1}(M)$ respectively, but {\it elliptic regularity} allows us to consider directly smooth forms (instead of $L^2$), whence equation \eqref{hodge}. Also, as stated our construction works for $M$ compact and closed (i.e., with no boundary). For either non-compact manifolds or for manifolds with boundary, one needs to restrict to the spaces of $L^2$ the $p$-forms with suitable conditions at infinity or at the boundary.}\end{footnote}}
\beq\label{hodge}\Omega ^p(M)= d\Omega ^{p-1}(M)\oplus d^* \Omega ^{p+1}(M)\oplus \mathcal H _p,\eeq
where $ \mathcal H _p= \{ \omega \in \Omega ^p(M): \; \Delta \omega =0\}=  \{ \omega \in \Omega ^p(M): \; d \omega =0\text{ and } d^*\omega =0\}.$
Clearly if $\Delta \omega=0$ or $\omega \in d\Omega ^{p-1}(M)$ (so that $d\omega=0$ in both cases), by taking $d$ on both sides of the diffusion equation and defining $\eta$ (i.e., Eq. \eqref{diffeq}), one obtains that $d\eta =0$, whence $\Delta \eta = dd^*\eta$. Therefore, one can take $\beta=\eta$ on $ d\Omega ^{p-1}(M)\oplus \mathcal H _p$. On the other hand, if $\omega \in d^* \Omega ^{p+1}(M)$, i.e. $\omega = d^* \alpha$ with $\alpha \in  d^* \Omega ^{p+1}(M)$ ( so that $d^* \alpha =0$), then $dd^*\omega =0$ and clearly taking $\beta$ constantly equal to $\omega$ solves Eq. \eqref{dd*diffusion} ( thereby proving that $(dd^*) ^a \omega=0$ in this case). In all the cases, we have shown the solution to Eq. \eqref{dd*diffusion} exists for every $t>0$.

The proofs that $d ^* (d^*d)^a =0$, $(d^*d)^b(d^*d)^a =0$ and that $(d^*d)^ad =0$ and $(dd^*)^a d^*=0 $ are analogous.
\end{proof}
\subsection{ Fractional differential}
One of the crucial objects of cohomology theory (essential in geometric quantization) is the notion of the differential of forms, which we generalize here to the fractional differential on a form as follows.
\begin{defn}
Let $\gamma\in (0,1)$. We define the {\it fractional differential } $d_\gamma$ via
\beq \label{frac-diff}
d_\gamma \omega = \frac{1}{2}\left( d\Delta ^\frac{\gamma-1}{2}\omega+ \Delta ^\frac{\gamma-1}{2}d \omega\right).
\eeq
\end{defn}
\noindent
where we recall the definition of $\Delta ^{-s}$ for for $s>0$ from eq. \eqref{fractionalintegral}
\beq
\Delta ^{-s} \omega = \frac{1}{\Gamma(s)}\int_0^{+\infty} e^{-t\Delta} \omega \frac{dt}{t^{1-s}}
\eeq
A few Lemmas are useful here.
\begin{lemma}
The adjoint of $d_\gamma$, denoted by $d_\gamma^*$, is given by
\beq	d_\gamma^*= \frac{1}{2}\left( d^*\Delta ^\frac{\gamma-1}{2}\omega+ \Delta ^\frac{\gamma-1}{2}d^* \omega\right). \eeq
\end{lemma}
\begin{proof}
By definition, the (formal) adjoint of $d_\gamma$ has to satisfy
$$\int _M\, d_\gamma \omega \wedge \alpha = \int _M\,  \omega \wedge d_\gamma^* \alpha. $$
The definition of $d_\gamma$ (Eq. \eqref{frac-diff}) implies that
$$\begin{aligned}  \int _M\, d_\gamma \omega \wedge \alpha &= \int _M\,   \frac{1}{2}\left( d\Delta ^\frac{\gamma-1}{2}\omega+ \Delta ^\frac{\gamma-1}{2}d \omega\right) \wedge \alpha = \int _M\,   \frac{1}{2} \left( \Delta ^\frac{\gamma-1}{2}\omega\wedge d^* \omega + d \omega\wedge \Delta ^\frac{\gamma-1}{2}\alpha \right)\\&= \int _M\,   \frac{1}{2} \left( \omega\wedge \Delta ^\frac{\gamma-1}{2}d^* \omega + \omega\wedge d^* \Delta ^\frac{\gamma-1}{2}\alpha \right), \end{aligned}$$
where we have used that $d^*$ is the adjoint of $d$ and the integration by parts formula for the Laplacian{\begin{footnote}{ The integration by parts formula for the fractional Laplacian on forms is easily proven appealing to the standard integration by parts formula for the fractional Laplacian on functions, aided by a partition of unity argument.}\end{footnote}}
$$\int _M \, \Delta ^b \eta \wedge \beta = \int _M \,  \eta \wedge \Delta ^b\beta,$$
for forms $\eta$ and $\beta$. Whence, reading from the first to last the equalities, we have
$$\int _M\, d_\gamma \omega \wedge \alpha \, dV=\int _M\,   \frac{1}{2} \left( \omega\wedge \Delta ^\frac{\gamma-1}{2}d^* \omega + \omega\wedge d^* \Delta ^\frac{\gamma-1}{2}\alpha \right), $$
which proves the Lemma.
\end{proof}
In fact, through the use of Proposition \ref{identities}, we can simplify the expressions of $d_\gamma$ and $d_\gamma^*$.
\begin{lemma}\label{simplerda}
$$d_\gamma = \frac{1}{2}\left( d(d^*d) ^\frac{\gamma-1}{2}\omega+ (dd^*) ^\frac{\gamma-1}{2}d \omega\right)$$
and 
$$d_\gamma ^*= \frac{1}{2}\left( d^*(dd^*) ^\frac{\gamma-1}{2}\omega+ (d^*d) ^\frac{\gamma-1}{2}d^* \omega\right).$$
\end{lemma}
\begin{proof}
Straightforward, using Proposition \ref{identities}.
\end{proof}
\noindent
Clearly $d_\gamma :\Omega ^p \to \Omega ^{p+1}$, i.e., if $\omega$ is a $p$-form, then $d_\gamma\omega$ is a $p+1$-form. Therefore, the question as to whether $d_a$ generates a complex springs to mind. This is answered in the affirmative by the proof of the following Proposition.
\begin{prop}

$$d _\gamma \circ d_\gamma=0 \qquad \text { and }\qquad d _\gamma^* \circ d_\gamma^*=0.  $$
\end{prop}
\begin{proof}
Using Lemma \ref{simplerda} and Proposition \ref{identities}, one immediately obtains
$$\begin{aligned} &d_\gamma d_\gamma \omega= \frac{1}{4}\left( d(d^*d) ^\frac{\gamma-1}{2}+ (dd^*) ^\frac{\gamma-1}{2}d \right)\left( d(d^*d) ^\frac{\gamma-1}{2}\omega+ (dd^*) ^\frac{\gamma-1}{2}d \omega\right) \\& =  \frac{1}{4}\left(  d(d^*d) ^\frac{\gamma-1}{2} d(d^*d) ^\frac{\gamma-1}{2}\omega + d(d^*d) ^\frac{\gamma-1}{2} (dd^*) ^\frac{\gamma-1}{2}d \omega+    (dd^*) ^\frac{\gamma-1}{2}d d(d^*d) ^\frac{\gamma-1}{2}\omega   +  (dd^*) ^\frac{\gamma-1}{2} d(dd^*) ^\frac{\gamma-1}{2}d\right)\\& =0.
\end{aligned}$$
\end{proof}
The fact that $d_\gamma$ forms a complex leads to the definition of the fractional cohomology groups $H^p_a(M,\mathbb R)= {\rm ker }\, d_\gamma /{\rm Im }\, d_\gamma$. The fractional Hodge theorem and the connection of these cohomology groups with the standard DeRham cohomology will be discussed in a future paper.
Perhaps more importantly for the purposes of this article is the fact that $d_a$ is the right fractional differential of forms in relation to the Fractional Laplacian of forms, as testified by the fact (which we shall prove in Theorem \ref{frac-lap-diff}) that $d_\gamma d_\gamma^*+ d_\gamma ^* d_\gamma$ is the fractional Laplacian.
In order to prove this fact, we first need to establish a few results.
\begin{prop}\label{commutation} For any form $\omega$ and any $b\in \mathbb R$, one has
\beq \label{eq:commutation1}  
d^* (d d^*)^b= (d^* d) ^{b}d^* \qquad d (d^*d)^b= (dd^*) ^{b}d,
\eeq
whence
\beq \label{eq:commutation2}  
\begin{aligned} &d^* (d d^*)^bd= (d^* d) ^{b+1}, \qquad d (d^*d)^bd^*= (dd^*) ^{b+1}, \qquad d (d^*d)^bd^*(dd^*)^c= (dd^* ) ^{b+c+1}\\& d ^*(dd^*)^bd(d^*d)^c= (d^* d) ^{b+c+1}.\end{aligned}
\eeq

\end{prop}

\begin{proof}

By definition (cf. Eq. \eqref{fractionalddstar}) one knows that{\begin{footnote}{Again, here one needs to argue that the solution to the diffusion equation for $d^*d$ has (a unique) solution. This is done in the same way as in Lemma \ref{identities} and specifically the arguments surrounding Eq. \eqref{diffeq}.}\end{footnote}}
\beq \label{fractionalddstard} (dd^*)^b  \omega= \frac{1}{\Gamma (-b)} \int _0^{\infty} \; \left( e^{-tdd^*} \omega - \omega\right) \frac{dt}{t^{1+b}},\eeq
where $\beta=e^{-tdd^*} \omega $ is the unique solution to

\beq\left\{ \begin{aligned}&\frac{\partial} {\partial t}  \beta + dd^* \beta = 0, \;\;\; \text{ for } \; (x,t) \in M\times \mathbb R_+\\& \beta (x,0) = \omega (x).\;\;\; \text{ for } \; x\in M\end{aligned} \right.\eeq
Taking $d^*$ of this equation produces
\beq\left\{ \begin{aligned}&\frac{\partial} {\partial t} (d^* \beta) + d^*d(d^* \beta) = 0, \;\;\; \text{ for } \; (x,t) \in M\times \mathbb R_+\\& d^*\beta (x,0) =d^* \omega (x),\;\;\; \text{ for } \; x\in M,\end{aligned} \right.\eeq
which shows that $d^*\left( e^{-tdd^*} \omega \right)=e^{-td^*d} d^*\omega$ and therefore calculating $d^*$ of both sides of Eq. \eqref{fractionalddstard}, we obtain
\beq d^* (dd^*)^b  \omega= \frac{1}{\Gamma (-b)} \int _0^{\infty} \; \left( e^{-td^*d} d^*\omega - d^*\omega\right) \frac{dt}{t^{1+b}}= (d^*d) ^b (d^*\omega).\eeq
The second formula in Eq. \eqref{eq:commutation1} is proven analogously and the formulae in Eq. \eqref{eq:commutation2} are consequences of Eq. \eqref{eq:commutation1}.
\end{proof}

\noindent
We can now make the following observation which readily follows from the Proposition above.
\begin{remark} Equations \eqref{eq:commutation1} in Proposition \ref{commutation} allow us to write the fractional differential as merely $d_\gamma= d\Delta ^\frac{\gamma-1}{2}$ or equivalently $d_\gamma = \Delta ^\frac{\gamma-1}{2}d$ and the analogous statements for $d_\gamma^*$.\end{remark}

We are now ready to prove
\begin{thm}\label{frac-lap-diff}
$$d_\gamma ^* d_\gamma+ d_\gamma d_\gamma^*= \left(\Delta \right)^\gamma.$$
\end{thm}
\begin{proof}
This is a simple calculation aided by Proposition \ref{commutation}.

\end{proof}

We also show that the definition is consistent with the standard Hodge Laplacian
\begin{thm}
One has that 
\beq
\lim_{\gamma \to 1^-} \Delta ^\gamma = \Delta
\eeq
and
\beq
\lim_{\gamma \to 0^+} \Delta ^\gamma = id.
\eeq
In particular, if $k\in \mathbb N$ then 
\beq 
\lim _{\gamma \to k} \Delta ^\gamma= \Delta ^k.
\eeq
\end{thm}
\begin{proof}
This is a simple consequence of the fact that
\beq \label{fraceucllap} (\Delta _x)^\gamma \alpha=\left( \left(-\sum _{m=1}^n \frac{\partial ^2 }{(\partial x^m)^2} \right) ^\gamma\alpha _{i_1\cdots i_p}\right)\, dx^{i_1} \wedge \cdots \wedge dx^{i_p}.\eeq
and that the result is well known for functions.
\end{proof}
\subsection{Fractional Curvature}
We now use the definition of the fractional differential $d_a$ to define the fractional curvature of a connection on a principal bundle $\mathcal P \to M$. Let $G$ be the Lie group of the bundle and $\mathfrak g$ its Lie algebra. We fix an open covering $\{ U_i\}$ of $M$.
More precisely,  given a connection $D$, we write it locally on each $U_i$ as $D=d+A_i$ (as customary), where $A_i$ is a $Lie(G)$-valued $1$-form, and then define the fractional connection to be the one modeled after $d_a$ such that,
$$D_a\phi = (d + A_i) \Delta ^{\frac{a-1}{2}} \phi= d_a \phi + A_i \Delta ^{\frac{a-1}{2}} \phi.$$
Given a connection $\alpha$, we will also denote its corresponding covariant fractional differentiation by
$$D_{a,\alpha}= (d+\alpha) \Delta ^{\frac{a-1}{2}}. $$
We then define the curvature to be,
$$F_{D_a}=D_{1,\Delta^{\frac{a-1}{2}}A}\circ D_{a,\Delta^{\frac{a-1}{2}}A},$$
which one readily reduces to
$$F_{D_a}=d_aA+[\Delta^{\frac{a-1}{2}}A,\Delta^{\frac{a-1}{2}}A].$$
Here, writing $A= A_j dx^j$ in local coordinates, with $A_j$ elements of the Lie algebra of $G$, as customary, one denotes
$$  [A,A]= [A_i, A_j] \; dx^i \wedge dx^j,$$
where $[A_i, A_j]$ is the commutator.
The {\it Gauge} group is given by sections of $Aut (\mathcal P)$ which act on fractional sections via
$$s^* D_a= s^{-1}\circ D_{1,\alpha}\circ s.$$

In the case in which $G=U(1)$, then
$$F_{D_a}= d_a A$$
and the Gauge group (identified as the sections of the form $s=e^\Lambda$) acts as
$$A\to A+d_a \Lambda.$$

\section{ Brief remarks on forms on manifolds with boundary}
Let $M$ be a manifold with smooth boundary $\partial M$. Let $\nu$ the unit normal to $\partial M$.

The natural generalization of the Dirichlert condition to forms is given by
\beq [Relative \; B.C.]\label{relBC} \left\{  \begin{aligned} &\omega \mid _{\partial M} =0\\& d^*\omega \mid _{\partial M} =0 \end{aligned}\right.
 \eeq
 and the generalization of the Neumann condition is
 \beq [Absolute\; B.C.]\label{absBC} \left\{  \begin{aligned} &{\rm i} _\nu \omega \mid _{\partial M} =0\\& {\rm i} _\nu  d\omega \mid _{\partial M} =0. \end{aligned}\right.
 \eeq
 Here, given any vector field $V$ and a $p$-form $\omega$, ${\rm i} _V \omega$ indicates the $(p-1)$-form determined by 
 $${\rm i} _V \omega (X_1,\cdots, X_{p-1})= \omega (X_1,\cdots, X_{p-1}, V),$$
 for arbitrary vector fields $X_1,\cdots, X_{p-1}.$
 It is a standard fact that for either of these boundary conditions, the integration by parts (or Green's formula) holds
 \beq \label{green} (\Delta \omega , \alpha ) = (d\omega, d\alpha) +  (d^*\omega, d^*\alpha).\eeq
 It is also a well known fact that the Hodge star operator interchanges these two conditions.

\section{The extension Theorem}

Here we show that there is a form of the Caffarelli-Silvestre theorem that works also for fractional Laplacians on forms.  Specifically we show
\begin{thm}(Caffarelli-Silvestre for forms)\label{CS}
Let $\omega$ a $p$-form in $\mathbb R^n$ which is in the domain the Hodge Laplacian $\Delta_x= d_xd_x^*+d_x^*d_x: \Omega ^p (\mathbb R^n)\to  \Omega ^p (\mathbb R^n)$. 
Let $\alpha \in \Omega ^p (\mathbb R^n\times \mathbb R_+) $ be a bounded solution to the extension problem
\beq\label{formsCS}
\left\{ \begin{aligned}& d(y^ad^*\alpha)+d^*(y^a d\alpha)=0 \text{ in }\; \mathbb R ^n\times \mathbb R_+\\& \alpha \mid _{\partial M} = \omega \text{ and } d^* \alpha \mid _{\partial M} =d_x^* \omega.\end{aligned} \right.
\eeq

\noindent 
with $a\in (-1,1)$. Then
\beq
\lim _{y\to 0} y^{a} {\rm i } _{\nu} d\alpha=C_{n,a} (\Delta_x) ^\gamma \omega,
\eeq
for some positive constant $C_{n,a}$, with $\nu = \frac{\partial}{\partial y}$ and $2\gamma= 1-a$.
\end{thm}
\begin{proof}
The moral of the strategy would be to make use the fact that we have shown that in $\mathbb R^n$ the fractional Laplacian on forms is given by
\beq \label{fraceucllap} (\Delta _x)^\gamma \alpha=\left( \left(-\sum _{m=1}^n \frac{\partial ^2 }{(\partial x^m)^2} \right) ^\gamma\alpha _{i_1\cdots i_p}\right)\, dx^{i_1} \wedge \cdots \wedge dx^{i_p}.\eeq
The way in which we implement the proof is to show, more directly, that the equation on forms reduces to Caffarelli and Silvestre equations on components.

In order to proceed, we choose coordinates $x^1,\cdots x^n, y$ on $\mathbb R^n\times \mathbb R_+$.
Also, if we write $\alpha=\alpha _{i_1\cdots i_p}dx^{i_1} \wedge \cdots \wedge dx^{i_p}+\alpha _{0 \ell _1, \cdots \ell _{p-1}} dy\wedge  dx^{\ell_1} \wedge \cdots \wedge dx^{\ell_{p-1}}$, a straightforward calculation yields
$$d\alpha=\frac{\partial \alpha _{i_1\cdots i_p} }{\partial y} \, dy\wedge dx^{i_1}\wedge \cdots \wedge dx^{i_p}+\sum _{k=1}^p \frac{\partial \alpha _{0 \ell _1, \cdots \ell _{p-1}} }{\partial x^{j_k}} dy\wedge  dx^{\ell_1} \wedge \cdots \wedge dx^{\ell_{p-1}}\wedge dx^{j_k}.$$
hence
$${\rm i } _{\nu} d\alpha=\frac{\partial \alpha _{i_1\cdots i_p} }{\partial y} \, dx^{i_1}\wedge \cdots \wedge dx^{i_p}+\sum _{k=1}^p \frac{\partial \alpha _{0 \ell _1, \cdots \ell _{p-1}} }{\partial x^{j_k}}  dx^{\ell_1} \wedge \cdots \wedge dx^{\ell_{p-1}}\wedge dx^{j_k}.$$
\noindent
This shows that in order to prove the theorem, we merely need to show that 
$$\lim_{y\rightarrow 0} \; y^a \frac{\partial \alpha _{i_1\cdots i_p}} {\partial y} =C_{n,a}(-\Delta)^\gamma \omega _{i_1\cdots i_p},
\qquad \text{ and } \qquad \lim _{y\to 0} y^{a} \frac{\partial \alpha _{0 \ell _1, \cdots \ell _{p-1}} }{\partial x^{j_k}} =0.$$

\noindent
For the purposes of the next few calculations we set $y= x^{n+1}$ so as to make the notation less cumbersome. From equation \eqref{eucld*} it follows immediately that

\beq \label{euclCSdd*}
\begin{aligned} d(y^ad^*\alpha)&= \sum _{\ell =1}^p (-1) \frac{\partial }{\partial x^{i_\ell}} \left( y^a \,\frac{\partial  \alpha_{i_1\cdots i_p} }{\partial x^{i_\ell}}\right) \, dx^{i_1} \wedge \cdots dx^{i_p}\\&+ \sum _{\ell =1}^p  \sum _{k =1}^{n+1-p}(-1)^{\ell}    \frac{\partial }{\partial  x^{j_k}} \left( y^a \,\frac{\partial  \alpha_{i_1\cdots i_p} }{\partial x^{i_\ell}}\right) \,dx^{j_k} \wedge dx^{i_1} \wedge \cdots \wedge \hat { dx^{i_\ell} } \wedge \cdots\wedge dx^{i_p}. \end{aligned}\eeq

and
\beq \label{euclCSd*d}
\begin{aligned} d^*(y^ad\alpha)&= \sum _{k=1}^{n+1-p}(-1) \frac{\partial }{\partial x^{i_\ell}} \left( y^a \,\frac{\partial  \alpha_{i_1\cdots i_p} }{\partial x^{i_\ell}}\right) \, dx^{i_1} \wedge \cdots dx^{i_p}\\&+ \sum _{\ell =1}^p  \sum _{k =1}^{n+1-p}(-1)^{\ell +1}    \frac{\partial }{\partial  x^{j_k}} \left( y^a \,\frac{\partial  \alpha_{i_1\cdots i_p} }{\partial x^{i_\ell}}\right) \,
dx^{j_k} \wedge dx^{i_1} \wedge \cdots \wedge \hat { dx^{i_\ell} } \wedge \cdots\wedge dx^{i_p}. \end{aligned}\eeq
Putting equations \eqref{euclCSdd*} and \eqref{euclCSd*d} together yields
\beq\label{euclCSlap}
d(y^ad^*\alpha)+ d^*(y^ad\alpha)=- \sum _{\ell =1}^{n+1} (-1) \frac{\partial }{\partial x^{i_\ell}} \left( y^a \,\frac{\partial  \alpha_{i_1\cdots i_p} }{\partial x^{i_\ell}}\right) \, dx^{i_1} \wedge \cdots dx^{i_p}.\
\eeq
Observing that the righthand side of equation \eqref{euclCSlap} is none other than ${\rm div} (y^a \nabla \alpha _{i_1\cdots i_p}) \, dx^{i_1} \wedge \cdots dx^{i_p}$ we can then write equation \eqref{formsCS} as
\beq\label{formsCS2}
\left\{ \begin{aligned}& {\rm div} (y^a \nabla \alpha _{i_1\cdots i_p})=0\in M\times \mathbb R_+\\& \left( \alpha _{i_1\cdots i_p}\right)\ \mid _{\partial M} = \omega _{i_1\cdots i_p}\ \text{ and } d^* \alpha \mid _{\partial M} =d_x^*\omega.\end{aligned} \right.
\eeq
Therefore, using the CS theorem, we have that 
\beq
\lim_{y\rightarrow 0} \; y^a \frac{\partial \alpha _{i_1\cdots i_p}} {\partial y} =C_{n,a}(-\Delta)^\a \ \omega _{i_1\cdots i_p},
\eeq
which proves that
$$\lim _{y\to 0} y^{a} {\rm i } _{\nu} d\alpha=(\Delta) ^a \omega,$$
since by (elliptic) regularity of solutions to equation \eqref{formsCS1}
$$\lim _{y\to 0} y^{a} \frac{\partial \alpha _{0 \ell _1, \cdots \ell _{p-1}} }{\partial x^{j_k}} =0.$$

Another (more invariant) way to proceed is as follows.
We decompose the equations for $\alpha$ by separating out (aided by Kodaira's decomposition) the form $\alpha$ into two parts $\alpha = \alpha _1$ and $\alpha _2$ such that $d\alpha _1=0$ and $d^*\alpha _2=0$.  We now focus on the equations that $\alpha _2$ satisfies. $\alpha _1=0$ and by abuse of notation, we write $\alpha$ for $\alpha_2$. So, we can first make the assumption that $d_x^*\omega=0$, which because of the Dirichlet boundary condition implies that $d^* \alpha \mid _{\partial M} =0$.
We make now the observation that from Eqs. \eqref{eucldd*} and \eqref{eucld*d}, it follows that
\beq\label{formsCS1}
\left\{ \begin{aligned}& d^*(y^a d\alpha)=0\in M\times \mathbb R_+\\& \alpha \mid _{\partial M} = \omega \text{ and } d^* \alpha \mid _{\partial M} =0,\end{aligned} \right.
\eeq
with $d^*\omega=0$.
In this case, using Eq. \eqref{fraceucllap}, we can rewrite Eq. \eqref{formsCS1} as
\beq
\left\{ \begin{aligned}& {\rm div} (y^a \nabla \alpha _{i_1\cdots i_p})=0\in M\times \mathbb R_+\\& \left( \alpha _{i_1\cdots i_p}\right)\ \mid _{\partial M} = \omega _{i_1\cdots i_p}\ \text{ and } d^* \alpha \mid _{\partial M} =0.\end{aligned} \right.
\eeq
The rest is as above.
\end{proof}
\begin{remark}\label{integer}
We can extend the theorem above to include $\gamma \in (0,\frac{n}{2})\setminus \mathbb N$, the argument is the same as the one in Chang-Gonz\'alez in  \cite{chang:2010} and we omit it. Also we observe that for $a=0$ one has that $\gamma =\frac{1}{2}$. This is will correspond to the absence of dilaton in the holographic theory.
\end{remark}
\section{Holographic scaling}
Holographic theories have been successful at modeling strongly coupled boundary theories via low energy approximations of string theories. 
There are basically two (equivalent) formulations of the AdS/CFT correspondence, due to  Banks,  Douglas,  Horowitz, and Martinec (BDHM) \cite{bdhm} and Witten,  Gubser, Klebanov, and Polyakov (GKPW) \cite{klebanov,klebanov2,Witten1998}. The AdS/CFT correspondence relates a string theory in its weak semiclassical limit defined on a conformally compact manifold $M$, called the {\it bulk}, to a conformal field theory (CFT) defined on the conformal boundary. Formally, the BDHM formulation prescribes that an operator $\mathcal O$ in the (conformal) boundary CFT is sourced by a field $\phi$ in the bulk and the ~n-point functions of $\mathcal O$ are determined via
\beq 
\label{correlators} \langle \mathcal O (x_1) \cdots \mathcal O (x_n)\rangle_{CFT} = \lim _{y\to 0} y^{-n\Delta} \langle \phi(x_1, y) \cdots \phi (x_n,y)\rangle_{bulk}.
\eeq
By Wightman's reconstruction theorem, the limit in Eq. \eqref{correlators} determines completely the form of the dual operator ${\mathcal O}$, at least formally. In the GKPW prescription, formulated in terms of path integrals, one assigns $\phi= y ^\Delta \phi _0$ in the bulk for some field $\phi_0$ in the boundary CFT (of which $\mathcal O$ is dual in the CFT sense),
\beq\label{corresp}
\langle e^{\int_{S^d}\phi_0{\cal O}}\rangle_{\rm CFT}=Z_S(\phi_0),
\eeq
thereby signifying that the boundary operator $\mathcal O$ plays the role of a current in the boundary CFT.

Applying this framework to condensed matter problems requires a  finite density of charge carriers, which holographically implies the existence of a conserved global charge and a (massless) bulk gauge field. 
One of the central points of the so called Effective Holographic Theory (EHT for short) proposed in \cite{skenderis} is to truncate the string theory to a finite spectrum of low-lying states, which amounts essentially to integrating out massive string modes.
We recall that in \cite{g1} and \cite{g2}, one calculates the effective holographic theories in order to study the IR regime of boundary strongly-coupled theories. The significant part of the effective action is then
\beq S=\int \mathrm{d}^{d+1}xdy\sqrt{-g}\left[\mathcal R-\frac{\partial\phi^2}2-\frac{Z(\phi)}4F^2+V(\phi)\right],\eeq
where the quantities $Z$ and $V$ are taken to have the asymptotics
\beq
\left\{\begin{array}{c}Z(\phi)\underset{\phi\to\infty}{\to}Z_0 e^{\gamma\phi} \\V(\phi)\underset{\phi\to\infty}{\to}V_0 e^{-\delta\phi}.\end{array}\right.
\eeq
The special but all-telling case 
\beq
\left\{\begin{array}{c}Z(\phi)=Z_0 e^{\gamma\phi}\\V(\phi)=V_0 e^{-\delta\phi}\end{array}\right.
\eeq
yields the following field equations
\beq
&R_{\mu\nu} + \frac{Z}{2} \, F_{\mu\rho} F^{\rho}_{\;\;\nu} -  \half \partial _\mu \phi \, \partial_\nu \phi \nonumber \\
&\qquad\qquad +\frac{g_{\mu\nu}}{2} \left[\half(\partial \phi)^2 - V -R   +  \frac{Z}{4}  F^2  \right] = 0 , \;\;\; \\
\label{scalar0}
&\Box \phi   = \frac{1}{4} Z'(\phi) \, F^2 +2-V'(\phi)  \, , \\
&\frac{1}{\sqrt{-g}} \, \partial_\mu \left(\sqrt{-g} \, Z(\phi) \, F^{\mu\nu}\right) =0  \, .
\eeq
We concentrate on Maxwell's equations,
\beq\label{maxwell}
\frac{1}{\sqrt{-g}} \, \partial_\mu \left(\sqrt{-g} \, Z(\phi) \, F^{\mu\nu}\right) = 0.
\eeq
One class of solutions found in \cite{g2} is
\beq\label{MassiveEMDsol}
\begin{split}
&d s^2=y^{\frac2d\theta}\left(\frac{L^2 d y^2+d R_{(d)}^2}{y^2}-\frac{d t^2}{y^{2z}}\right),\qquad  A=Q_0y^{\zeta-\xi-z} d t\,,\qquad  e^\phi=y^{\pm\kappa}\,,\\
& L^2V_0=(d-1+z-\theta)(d+z-\theta)+(z-1)\xi\,,\qquad  Q_0^2= \frac{2 (z-1)}{Z_0(z-\zeta+\xi )}\,,\\
&\kappa=\sqrt{2 (1-z) (\zeta-\xi) +\frac2d\theta(\theta-d) }\,,\qquad \delta\kappa =\pm \frac{2\theta }{d}\\
& \gamma\kappa = \pm2\left(\frac1d\theta- \zeta +\xi\right)\,,\qquad \epsilon=\gamma-\delta\,.
\end{split}
\eeq

Significantly in all the solutions found in \cite{g1,g2,skenderis}, the dilaton is such that $e^\phi=y^{\pm\kappa}$, which as remarked in the introduction is identical to the Domokos/Gababadze\cite{domokos} mechanism.
Consequently, the main application of this note is incarnated in the fact that a consequence of Theorem \ref{CS} is that this action induces the fractional Maxwell equations at the boundary
$$\Delta ^a A^t= 0,$$
where $A^t$ is the tangential (boundary) component and $a$ is determined by $\kappa$. 
This is of course a straightforward application of the afore-mentioned theorem, given that Eq. \eqref{maxwell} is of the form of Theorem \ref{CS}. In fact, from the expressions in equation \eqref{MassiveEMDsol} above, one sees that $\sqrt{-g} \, Z(\phi)= y^a$ (up to constant factors) and therefore the vacuum Maxwell equations (eq. \eqref{maxwell}) reduce to
\beq
 \partial_\mu \left(y^a \, F^{\mu\nu}\right) = 0,
\eeq
or in differential form
\beq
d(y^adA)=0,
\eeq
where $a= d+1+\theta-z\pm\gamma \kappa$.
\begin{remark}
As observed in Remark \ref{integer}, even if $a=0$, one gets a power of $\frac{1}{2}$ in the Laplacian.
\end{remark}
Consequently, the dimension of the gauge field at the boundary is indeed non-traditional: $[A^t]=a$.  Analogously the associated current also has an anomalous dimension: $[J]=d-1-a$.  We see clearly then that it is the Caffarelli/Silvestre mechanism that accounts for the generation of anomalous dimensions in the boundary  ``gauge'' theory in holographic constructions.  When dilatons are absent from the bulk, that is $a=1$, the standard result obtains\cite{Witten1998,klebanov,maldacena} in which anomalous dimensions are absent.  In this regard, it is a bit of a misnomer to refer to the boundary theory as acquiring an anomalous dimension because the dilaton field in the bulk gives an effective dimension to the the bulk gauge field of dimension $(1-a/2)$.  Hence, it is not as if the boundary theory acquires an anomalous dimension from renormalization.  Its non-traditional dimension is fixed from the dilaton dynamics in the bulk.  The effective running of the field strength in the bulk as dictated by the dilaton coupling, $y^a$, manifests itself as a non-locality in the boundary response.
We remark that the boundary current $J$ (this is what we called $\mathcal O$ in the beginning of this section, we are changing the name to $J$ as it is more consistent with the standard Maxwell equations) corresponding to the field sourced by $A^t$ satisfies, by our main theorem satisfies
\beq 
(\Delta )^\gamma A ^t= J
\eeq
thus identifying the boundary theory as a fraction EM theory.
\section{Gauge group and currents }

In the standard $U(1)$ gauge theory, the {\it local} action of the gauge group is given by transforming the complex sections in the Dirac Lagrangian (or any other gauge theory) via $\phi \to e^{i\xi}\phi$ for some (real) function $\xi$. Such transformations are the correct ones in terms of the covariant derivative $D,$ locally given by $d+i qA,$ in that $D\phi$ transforms as $e^{i\xi} D\phi$, provided that $A$ transforms as $A_\mu \to A_\mu +\partial _\mu \xi$. 
A {\it global} symmetry would be given by a similar transformation $\phi \to e^{i\rho}\phi$, where $\rho \in \mathbb R$ is a real number. This symmetry would clearly leave the Lagrangian unchanged, and just change the phase of $\phi$ by a constant factor. It would also leave the connection unchanged (and as a consequence not change any gauge). A local transformation of the form $e^{iq\xi}\phi $ leaves  unchanged the Dirac Lagrangian
$$ \mathcal L= - \overline {\phi} {\slashed {D}_{A}}\phi - m\phi \overline \phi.$$
If $A$ were a classical $U(1)$ Gauge field, then $\phi \to \phi '= e^{iq\xi}\phi $ is exactly what one needs and in order for the Lagrangian to be unchanged under 
$A\to A' = A+ d\xi$ as the Gauge transformation on the connection 1-forms $A$. If $A$ is a fractional field though, we need to take a local transformation on matter that respects the fractional Gauge transformation,
$$A\to A'= A+d_a\xi.$$
In order to achieve this, since $d_a= d \, (\Delta)^{\frac{a-1}{2}}$, we simply take as a local transformation
\beq \phi \to \phi '= e^{i q (-\Delta)^{\frac{a-1}{2}}\xi} \phi.
\eeq
Since $D_A= d-iq A$, one readily calculates that with this local action, $D_A$ transforms under such local transformations as
\beq D_A\to  e^{i q (-\Delta)^{\frac{a-1}{2}}\xi}  D_{A'},
\eeq
where $A'= A+d_a\xi$, thus leaving the Dirac Lagrangian unchanged. This is readily seen by the calculation
\beq (d+A') \phi'=  (d+A') \left(  e^{i q (-\Delta)^{\frac{a-1}{2}}\xi}\phi\right)=  e^{i q (-\Delta)^{\frac{a-1}{2}}\xi}\left( d\phi +i q d_a \xi \, \phi + A'\phi \right).
\eeq

We can turn this into the discussion which essentially appeared in \cite{tor} by considering the 1-form $\alpha$ (which is $a_\mu$ in \cite{tor}) defined by
\beq
\Delta ^a \alpha = A
\eeq
and then consider the Lagrangian
$$ \mathcal L= - \overline {\psi} {\slashed {D}_{\alpha}}\psi - m\psi \overline \psi,$$
where now $D_\alpha =  d-iq \alpha$, and $q$ is unitless. Then, if $\Lambda$ is defined through
\beq
\Delta ^{\frac{a-1}{2}} \xi = \Lambda,
\eeq
we have that $\alpha$ transforms after the local action $\psi \to e^{iq\Lambda}\psi$, as
$$\alpha \to \alpha '= \alpha -i q d\Lambda.$$

On a related note, the dual fields in the boundary theory\cite{klebanov,Witten1998} are currents. As such, these are differential forms with no {\it local} gauge group action.   What we have shown in this paper is that the currents in the boundary theory for a bulk theory containing a dilaton coupling are currents generated by a fractional ``gauge'' theory.

\section{Final Remarks}

We have developed the notion of fractional differentiation of p-forms via the fractional Laplacian.  Such an operator naturally appears anytime elliptic differential equations are recast in a spacetime with one lower dimension.  What we have shown here is that the same is true for p-forms.  Since all holographic constructions\cite{g1,g2} to date result in equations of motion that are identical to those that underlie the p-form generalization of the CS extension theorem, fractional Maxwell equations naturally result at the boundary.  This result then lends credence to work\cite{tor} which was based on the intuition that anomalous dimensions for gauge fields results in a non-local gauge-invariant condition.  This work lays plain the precise form of the non-locality involves the fractional Laplacian ala $A^t \rightarrow A^t+ d_\gamma \Lambda$ with 
 $d_\gamma \equiv (\Delta)^\frac{\gamma-1}{2}d$ rather than the standard fractional derivative.  That the boundary theory yielding an anomalous dimension must involve the fractional Laplacian is not unexpected since the Ward identities explicitly preclude anomalous dimensions from purely local gauge theories.  What this work demonstrates is that anomalous dimensions for gauge fields  are a signature that the quantum field theory is necessarily a boundary theory on some manifold.

\textit{Acknowledgments:} We thank Blaise Gouteraux and V. P. Nair for  useful remarks and the NSF DMR-1461952 for partial funding of this project.
%

\end{document}